\documentclass[%
 reprint,
 amsmath,amssymb,
 aps,
]{revtex4-1}

\usepackage{graphicx}
\usepackage{dcolumn}
\usepackage{bm}
\usepackage{amssymb}

\newtheorem{theorem}{Theorem}
\newtheorem{proof}{Proof}
\newtheorem{corollary}{Corollary}

\begin{document}


\title{Force-free Electrodynamics and Foliations in an arbitrary Spacetime}

\author{Govind Menon}
\affiliation{Department of Chemistry and Physics\\ Troy University, Troy, Al 36082}

\date{\today}

\begin{abstract}
In this paper we formulate the relationship between force-free electrodynamics and foliations. The background metric, is considered predetermined and electrically neutral, but otherwise arbitrary. As it turns out, solutions to force-free electrodynamics is intimately connected to the existence of foliations of a spacetime with prescribed properties. We also prove a local existence and uniqueness theorem and provide a recipe for constructing the unique solution/class of solutions when certain conditions are met. We clarify the theorem with examples. We are also able to also prove a singularity theorem for when non-null solutions approach the null limit. Here too, we construct an explicit example to illustrate the singularity theorem. For variety of discussion, we conclude with a solution to Maxwell's equations in FRW cosmology.
\end{abstract}

\pacs{Valid PACS appear here}
\maketitle


\section{INTRODUCTION}

Force-free electrodynamics (FFE) has become the central framework for describing the magnetospheres of active black holes. Although the governing equations of force-free electrodynamics were developed as early 1977 (\cite{BZ77}, \cite{Carter79}), it did not receive much attention till the late 90's as a subject of systematic study (\cite{Uchida1}, \cite{Uchida2}). More recently, the force-free magnetosphere in a Kerr background has played a prominent role in the study of black hole astrophysics. Properties of the force-free magnetosphere and its abilities to extract energy and angular momentum was a general feature of a numerical study of the subject ( \cite{TT14}, \cite{Rio18} and \cite{Tchekhovskoy_2012}). 
From a theoretical point of view, analytical solutions to a Kerr force-free magnetosphere slowly emerged as well (\cite{MD07}, \cite{BGJ13} and \cite{Menon15}). A recent paper by Gralla and Jacobson (\cite{GT14}) captures the current status of the theory of FFE.  

In this paper we will focus on the initial value problem of FFE such as it is. Recent work using Euler potentials has suggested that FFE is ill posed (\cite{RR17}). Others have concluded that FFE is deterministic in the magnetically dominated case (\cite{GT14}, \cite{PH13}, and \cite{Pal11}). We are able to write down a deterministic evolution equation in both the electrically and magnetically dominated case. The resulting solution is unique. We are also able to show that in the null case there will be a unique class of solutions. As we will show below, FFE is deterministic (in the generalised sense of class of solutions) regardless of sign and value of $F^2$. Here $F$ is the electromagnetic field tensor. Additionally, we will prove a local existence and uniqueness theorem for FFE. In \cite{CGL16}, the authors argue that the existence of a stationary, axis-symmetric, magnetically dominated, force-free electromagnetic field in a Kerr background is entirely dependent on the existence of a foliation of Kerr spacetime with certain well prescribed properties. In the following sections we will prove the same as a general result. In other words, we will elevate 2-dimensional foliations with certain well prescribed properties as a fundamental object for FFE in an arbitrary spacetime. These prescriptions are very natural and expected, and cannot be relaxed.  Geometry (or gravity if you prefer) alone determines the existence of a force-free electromagnetic field. We do not fix the background, nor will we have any restrictions of stationarity and axis-symmetry. The foliations we refer to are integral submanifolds of an involutive distribution that is the kernel of the 2-form $F$. It will also become clear that the notion of an initial value problem for FFE has to be replaced by the study of FFE admissible spacetimes and foliations.

Curiously, while our formalism is able to treat the general evolution equations in a cohesive manner, it is not the case that $F^2$ can change its sign smoothly. I.e., we first describe a null field, and then separately an electrically or magnetically dominated field. A smooth transition is not in general permitted. To be clear, it is not the failure of our formalism, but indeed additional restrictions are required to allow this transition if indeed such a transition is possible. In particular, we will prove that in the case of foliations of spacetime generated by commuting Killing vector fields, non-null solutions do not have a well defined null limit. The solution necessarily becomes unbounded in this case.   

We begin with a recapitulation of the basic equations and properties of FFE, following which we will recast the equations in a coordinate system that is adapted to the foliations. It is in this adapted chart that we will prove our existence and uniqueness theorems. To illustrate the computational ability of our formalism, we will re-derive the previously obtained solutions in \cite{MD07}, \cite{BGJ13} and \cite{Menon15} using the new formalism. Spacetimes containing commuting Killing vector fields are then treated as a special case, because, as we will show they necessarily give rise to a special FFE solution under very mild restrictions. As an example we will conclude with a pair of vacuum solutions in a Kerr background. All solutions presented in this paper are exact.

In this article, as usual, comma denotes partial derivatives, i.e.,  $$M_{,r} \equiv \frac{\partial M}{\partial x^r}\;.$$
Also, in $\sqrt{-g}$, $g$ denotes the determinant of the metric, i.e., $\det g$. Otherwise $g$ is simply the metric, and $\epsilon_{\mu\nu \alpha \beta }$ is the Levi-Civita tensor such that 
$\epsilon_{1234} = \sqrt{- g}\;.$
Notice the unusual 1-4 labelling of indices. This will be made clear along the way.

\section{Basic Properties of the Force-Free Electromagnetic field}
A spacetime, for our purposes, is a 4-dimensional smooth manifold ${\cal M}$ endowed with a metric $g$ of Lorentz signature, specifically $(-1, 1,1,1)$. In this work, we are concerned with the evolution of the electromagnetic field generated by a (possibly) non-trivial current density. For this reason, since we would like to for account for all the electromagnetic interactions in our formalism, we will assume that the background metric satisfies the Einstein equation for at most electrically  neutral sources. We do not place any further restrictions on $g$ for our main theorem. We will consider the background as fixed and electrically neutral (i.e., background $F=0$). Maxwell's equations in tensor form is usually written as
\begin{equation}
\nabla_\nu F^{\mu\nu}= J^\mu\;,
\label{Max_1}
\end{equation}
and
\begin{equation}
\nabla_{[\mu} F_{\nu\lambda]}=0\;.
\end{equation}
Here, $F$ is the Maxwell field tensor and $J$ is the current density. Also, $[\mu\nu\lambda]$ denotes anti-symmetrization of the included indices in the usual manner.

Following Gralla and Jacobson (\cite{GT14}), we will mainly use the formalism of exterior calculus to describe the electromagnetic interaction. Please see appendix \ref{ExtCalc} for details. In this case, $F$ is to be viewed as a closed 2-form, i.e.,
\begin{equation}
 d  F = 0\;,
 \label{Fclosed}
\end{equation}
 which satisfies
\begin{equation}
* \;d * F = J\;.
\label{inhomMaxform}
\end{equation}
Here $*$ is the Hodge-Star operator and $d$ is the exterior derivatives on forms. From the above equation and eq.(\ref{Max_1}), it is clear that we use $J$ to describe both the current ``vector" density or the associated 1-form. Recall that a vector field in geometry is also referred to as a contravariant vector field by physicists, and a 1-form, or a dual vector field is the usual covariant vector field. The distinction is made either by context or explicitly by denoting its index. For example, $J^\mu$ is a contravariant vector filed, while $J_\mu \equiv g_{\mu\nu} J^\nu$ is the associated covariant vector field. The same distinction is to be understood for all tensorial objects.
\vskip0.2in
An electromagnetic field is degenerate if there is a vector field $w$ such that the interior product of $w$ with $F$ vanishes. I.e.,
$$
i_w F \equiv F(w, \cdot)=0\;.
$$
Force-free electrodynamics is a special case where
the degenerate field satisfies
$$i_J\; F =0\;,$$
where $J$ is the current density given by eq.(\ref{Max_1}). In components, this can be written as
\begin{equation}
F_{\mu\nu}\nabla_\lambda F^{\nu\lambda} =0\;.
\label{FF v1}
\end{equation}
Our central focus will be to establish to a meaningful initial data set/surface for the above equation, and also to prove a local existence and uniqueness theorem.

\vskip0.2in
We begin our formulation of the problem by recalling a few properties of degenerate fields (/force-free electrodynamics). The interested reader is referred to  \cite{Carter79}, \cite{Uchida1}, and \cite{Uchida2} for details. All the essential properties of FFE has been recast in an efficient way using modern notation in \cite{GT14}, and should serve as a  reference guide for our current work.
A force-free electromagnetic field is a simple 2-form. I.e., there exists 1-forms $\alpha$ and $\beta$ such that
$$F= \alpha \wedge \beta\;.$$
Consequently, the kernel of $F$ is a 2 dimensional subspace of the tangent bundle consisting of all vector fields $v$ such that $\alpha(v)=0=\beta(v)$. We will then have that $i_v F =0$.

\vskip0.2in
Depending on the causal character of the kernel of $F$, denoted as $\ker F$, we can locally classify $F$ into three categories. If for any point $p$ in our spacetime, the kernel of $F$ at $p$ is  spacelike, Lorentz, or if the metric restricted  to the kernel is degenerate, we say that $F$ is electrically dominated, magnetically dominated or null at $p$. This is equivalent to the requirement that the scalar
$$F^2 (p)\equiv F_{\mu\nu} F^{\mu\nu} (p)$$
is less than zero, greater than zero or zero respectively. Recall, that the metric restricted to $\ker F |_p$ is degenerate  if there exists $l \in \ker F |_p$ such that $g(l, v)=0$ for all $v \in \ker F |_p$. In particular $l$ is a null vector.

\vskip0.2in
From Cartan's magic formula, we get that
$$v \in \ker F\;\;{\rm implies \;that} \;\;{\cal L}_v F =0\;.$$
In the remainder of the section we will rely on terminology and key results on foliation from differential geometry. Please see appendix \ref{DGAppen} for a refresher on the definitions/results used below.
Now suppose $v, w \in \ker F$. Since
$$i_{[v,w]} F= [{\cal L}_v, i_w] F =0\;, $$
we have that $\ker F$ is an involutive distribution, and therefore by Frobenius' theorem, spacetime can be foliated by 2-dimensional integral submanifolds of the distribution spanned by $\ker F$.

\vskip0.2in
In particular, given any point $p \in {\cal M}$, there exists a coordinate chart $\big(U_p, \phi_p= (x^1, \dots, x^4)\big)$ centered about $p$, i.e.,
$$ \phi_p (p) = \big(x^1(p), \dots, x^4(p)\big)=(0, \dots, 0)\;,$$
that is adapted to the distribution. Without loss of generality, this means we can arrange for
\begin{equation}
 {\rm span}\;\left\{ \frac{\partial}{\partial x^1}, \frac{\partial}{\partial x^2}\right \} = \ker F\; \Big|_{U_p}\;.
 \label{adapchart}
 \end{equation}
We will take eq.(\ref{adapchart}) as part of the requirements for our {\it adapted chart} at $p$ for the foliation determined by $\ker F$.
There is no preference here for a timelike coordinate, and so we label the adapted coordinates with indices ranging from $1-4$, rather than the usual $0-3$. 
In this chart, the electromagnetic field tensor can be written as
$$F= u \;dx^3 \wedge dx^4$$ for some component function $u(x^1, \dots, x^4)$.  Eq.(\ref{Fclosed}) now limits $u$ to only be a function of $x^3$ and $x^4$. We then obtain the needed final form of $F$:
\begin{equation}
F= u(x^3, x^4) \;dx^3 \wedge dx^4\;.
\label{finalformF}
\end{equation}

\section{Equations of Force-Free Electrodynamics in an Adapted Coordinate System}
All our results will only be valid locally, and so we restrict all discussions to the adapted coordinate system described above. We will also have the occasion to require that the domain of the chart $U_p$ is {\it starlike} about $p$. This means that if $q \in U_p$, then the line segment from $p$ to $q$ lies in $U_p$.

We seek an expression for $u$ in eq.(\ref {finalformF}) satisfying eq.(\ref{inhomMaxform}) such that
\begin{equation}
(J^\sharp)^a = 0 \;\;{\rm for}\;\;a=3,4\;.
\label{Jcond}
\end{equation}
Here, $\sharp$ is the raising operator defined by
$$(J^\sharp)^\mu = g^{\mu\nu} J_\nu\;.$$
This will ensure that the resulting current is force-free. Applying the Hodge-Star and the exterior derivative in the appropriate order to $F$ in eq.(\ref{finalformF}) gives that

$(J^\sharp)^a=\big((*\;d * F)^\sharp\big)^a$
\begin{equation}
= \frac{1}{2} \;\epsilon^{r \alpha \beta a} \;\partial_r \left(u\;\epsilon_{\mu\nu \alpha \beta }\; g^{\mu 3} \;g^{\nu 4}\right)\;. 
\label{jexpress}
\end{equation}
The requirements of eq.(\ref{Jcond}) now gives our equations in component form:
\begin{equation}
\epsilon^{r \alpha \beta a} \;\partial_r \left(u\;\epsilon_{\mu\nu \alpha \beta }\; g^{\mu 3} \;g^{\nu 4}\right)=0\;\;{\rm for} \;\;a=3,4 \;.
\label{FF v2}
\end{equation}
A solution for $u$ in the above equation will result in a force-free electromagnetic field or a trivial case of a vacuum solution where $J=0$. To simplify the above expression when $a=3$, define a quantity $M^r$ by 
$$M^r= g^{r 3} \;g^{3 4}- g^{3 3} \;g^{r 4}\;.$$
Naturally, despite its notational appearance, $M^r$ is not a vector field. Then with the help of eq.(\ref{levicontract}) it is easily seen that
$$
\epsilon^{r \alpha \beta 3}  \epsilon_{\mu\nu \alpha \beta }\; g^{\mu 3} \;g^{\nu 4}= -2 M^r\;.
$$
Also, by taking the derivative of $\sqrt{-g}^{\;-1}$ we see that
$$
\partial_r\; \epsilon^{r \alpha \beta a}= - (\partial_r \ln \sqrt{-g}) \;\epsilon^{r \alpha \beta a}\;.
$$
Finally, observing that
\begin{widetext}
$$\epsilon^{r \alpha \beta a} \;\partial_r \left(u\;\epsilon_{\mu\nu \alpha \beta }\; g^{\mu 3} \;g^{\nu 4}\right)=\epsilon^{r \alpha \beta a} \;\epsilon_{\mu\nu \alpha \beta }\; g^{\mu 3} \;g^{\nu 4}\; \;\partial_r u\;+\;u\;\partial_r \left(\epsilon^{r \alpha \beta a} \;\epsilon_{\mu\nu \alpha \beta }\; g^{\mu 3} \;g^{\nu 4}\right)-u\;\epsilon_{\mu\nu \alpha \beta }\; g^{\mu 3} \;g^{\nu 4}\;\partial_r \left(\epsilon^{r \alpha \beta a} \right)\;\;,$$
\end{widetext}
we get that when $a=3$, eq. (\ref{FF v2}) reduces to
$$M^4 \;\frac{\partial}{\partial{x^4}} \ln |u| = - \frac{1}{\sqrt{-g}}\; \frac{\partial}{\partial{x^r}} \left(\sqrt{- g}\; M^r\right)\;.$$
In exactly the same way, when $a=4$, eq. (\ref{FF v2}) reduces to
$$N^3 \;\frac{\partial}{\partial{x^3}} \ln |u| = - \frac{1}{\sqrt{-g}}\; \frac{\partial}{\partial{x^r}} \left(\sqrt{- g}\; N^r\right)\;,$$
where
$$N^r= g^{r 3} \;g^{4 4}- g^{3 4} \;g^{r 4}\;.$$
Although $M^r$ is not a 4-vector we will still write
$$\frac{1}{\sqrt{-g}}\; \frac{\partial}{\partial{x^r}} \left(\sqrt{- g}\; M^r\right)$$
as
$$\nabla_r M^r\;.$$
This strictly for notational convenience. Similar remarks apply to $N^r$ as well.
With this simplification, and setting
\begin{equation}
M\equiv M^4= -N^3\;,
\label{Mdef}
\end{equation}
we can now write the equations of force-free electrodynamics as
\begin{equation}
M\;\frac{\partial}{\partial{x^4}} \ln |u| =-\nabla_r\; M^r
\label{FF v3_a}
\end{equation}
and
\begin{equation}
M\;\frac{\partial}{\partial{x^3}} \ln |u| =\nabla_r\; N^r\;.
\label{FF v3_b}
\end{equation}
We have used the fact that $u$ is only a function of $x^3$ and $x^4$ in deriving the above equations. Eq.(\ref{FF v3_a}) and (\ref{FF v3_b}) will serve as our basic equations of FFE for the remainder of the paper.
\section{Initial Data, Local Existence and Uniqueness Theorem}
As mentioned earlier, in this section, we will prove two local existence and uniqueness theorem; one for the null case, and the other for when the field is  electrically or magnetically dominated. We will then supply examples for both cases. Both examples are previously obtained FFE solutions recast in the adapted chart formalism. The existence theorem in the null case is almost trivial.
\subsection{The Null Force-Free Field}
\label{nullcondsect}

When $M = 0$ there is a non-trivial solution to the equation
$$\left(
    \begin{array}{cc}
      g^{33} & g^{34} \\
      g^{43} & g^{44} \\
    \end{array}
  \right)\left(
           \begin{array}{c}
             \chi_3 \\
             \chi_4 \\
           \end{array}
         \right) =0\;.
$$
I.e., there exists a 1-form
$$\chi = \chi_3 \;dx^3 + \chi_4 \;dx^4$$
such that the vector field
$\chi^\sharp \in \ker F$, meaning
$$\chi^\sharp = \chi^1 \frac{\partial}{\partial x^1}+ \chi^2 \frac{\partial}{\partial x^2}\;.$$
On the other hand, we can always write $F$ in the form
\begin{equation}
 F = \bar u \;(\psi \wedge \chi)\;,
 \label{Finnullbasis}
\end{equation}
where $g(\psi,\chi)=0$ and $\bar u$ is a new component function. Then
$i_{\chi^\sharp}\; F = 0 $ implies that
$$\chi(\chi^\sharp) =0 = \psi(\chi)\;.$$
In this case, $\chi$ is a null vector and  $F^2 = 2 \bar u^2 \psi^2 \chi^2 =0$. Consequently, $F$ is a null force-free field. Moreover, since
$$g(\chi^\sharp, W)= \chi(W) =0$$
for every $W \in \ker F$, we have that the metric when restricted to the kernel of $F$ is degenerate as previously discussed.

We know that null force-free solutions exists, for example see \cite{MD07}, (and \cite{BGJ13} for its generalizations). So, from physical grounds, it is only expected that there should be well defined evolution equations in this case. That is almost true. It turns out that if null solutions exists, they come as a class of solutions with an arbitrary function of 2 variables. 
\begin{theorem}
Let ${\cal F}$ be a $2$-dimensional foliation of a spacetime with metric $g$.  Let $\big(U_p, \phi_p= (x^1, \dots, x^4)\big)$ be an adapted chart about any arbitrary point $p$. Suppose $M=0$ (as defnied in eq.(\ref{Mdef})) in $U_p$. Then $F$ given by eq.(\ref{finalformF}) for any smooth function $u(x^3,x^4)$ is a unique class of force-free solution satisfying eq.(\ref{adapchart}) in $U_p$ if and only if 
$$\nabla_r M^r = 0 = \nabla_r N^r\;.$$
\label{mainexistuninull}
\end{theorem}
\begin{proof}
This is a trivial consequence of eqs.(\ref{FF v3_a}) and (\ref{FF v3_b}).
$\blacksquare$
\end{proof}
The proof is almost too easy that it begs for an example.  To illustrate the previous theorem, we will now describe the foliations generated by a null electromagnetic field denoted by $F_{Null}$. This solution was first obtained as an exact solution to the Blandford-Znajek equations. For details regarding its original derivation see \cite{MD07}. The current density vector in this case is proportional to the in-falling principal null geodesic of the Kerr geometry, which in the Boyer-Lindquist coordinates takes the form:
$$n=\frac{(r^2+a^2)}{\Delta} \;\partial_t -\partial_r + \frac{a}{\Delta}\;\partial_\varphi\;.$$
We will explicitly construct an adapted coordinate system and show that equations eq.(\ref{FF v3_a}), and eq.(\ref{FF v3_b}) are not violated.  The kernel of $F_{Null}$ is given by (\cite{MD07})
$$\ker F_{Null} = {\rm span}\;\{\Delta\;n, a \sin^2(\theta)\; \partial_t + \; \partial_\varphi\}\;.$$
As per the theorem, the above kernel fixes the entire class of solution.
For computational ease we have picked $\Delta\;n$ as a basis vector for the foliation. Set
$$X_1 = \Delta n$$
and
$$X_2 = a \sin^2(\theta) \;\partial_t + \; \partial_\varphi\;.$$ 
Fix $Q(t,r)$ by the expression
$$-2\left(t +r+\frac{r_+ ^2+ a^2}{r_+-r_-}\ln|r-r_+|-\frac{r_- ^2+ a^2}{r_+-r_-}\ln|r-r_-|\right)\;.$$
The defining properties of $Q$ are 
\begin{equation}
    Q_{,t} = -2 \;{\rm and}\;Q_{,r} = -2\;\frac{(r^2+a^2)}{\Delta}\;.
    \label{Qprop}
\end{equation}
Further, define vector fields
$$X_3 = Q(t,r)\; \partial_t - \tan(\theta)\; \partial_\theta$$
and
$$X_4=\partial_\varphi\;.$$
It is easily verified that $[X_i, X_j]=0$ for $i,j=1,\dots,4\;.$ Therefore, there exists an adapted coordinate system $(x^1, \dots, x^4)$ such that
$$\frac{\partial}{\partial x^i}=X_i\;.$$
Since there is no real need, we do not calculate the transformation functions for $\{x^i\}$ in terms of $(t,r,\theta,\varphi)$ explicitly. It is however important to note that, by construction, eq.(\ref{adapchart}) is automatically satisfied.
The coordinate $1$-forms transforms as

\begin{widetext}
\begin{equation}
  \left(
    \begin{array}{c}
      dx^1 \\
      dx^2 \\
      dx^3\\
      dx^4
    \end{array}
  \right)=\left(
           \begin{array}{cccc}
             0 &\frac{-1}{\Delta}&0&0 \\
             \frac{1}{a \sin^2(\theta)}&\frac{r^2+a^2}{a \Delta \sin^2(\theta)}&\frac{Q}{a\sin^2(\theta) \tan(\theta)}&0 \\
             0&0&\frac{-1}{ \tan(\theta)}&0\\
             \frac{-1}{a \sin^2(\theta)}&\frac{-\rho^2}{a \Delta \sin^2(\theta)}&\frac{-Q}{a \sin^2(\theta)\tan(\theta)}&1
           \end{array}
         \right) \left(\begin{array}{c}
      dt \\
      dr \\
      d\theta\\
      d\varphi
    \end{array} \right)\;.
    \label{nulladapforms}
\end{equation}
    \end{widetext}
    To compute $\{M^r\}$ and $\{N^r\}$ in the adapted frame we need to first calculate the components of the inverse metric. This is easily done by noting that 
    $$g^{ij} = g(dx^i, dx^j)\;.$$
    So we apply the above expressions in eq.(\ref{nulladapforms}) to eq.(\ref{invgkerr}), and for example observe that
    $$g^{11}= g(dx^1, dx^1)=\frac{1}{\rho^2 \Delta}\;.$$
    Proceeding in a similar manner and using the definitions of $\{M^r\}$ and $\{N^r\}$ we get that the only non-vanishing components of $M^r$ and $N^r$ are
    $$M^1=-\frac{\cos^2(\theta)}{\sin^4(\theta)}\;\frac{1}{a \rho^2 \Delta}$$
    and
       $$N^1=\frac{ \cos^2(\theta)}{\sin^6(\theta)}\;\frac{Q}{a^2 \rho^2 \Delta}\;.$$
   In particular, as expected, here $M=0$, and our previous theorem applies. Once the components of the metric, or its inverse, are found, its determinant is easily computed as well, and we get that 
    $$\sqrt{-g}=a \rho^2 \sin^3 (\theta) \Delta \tan(\theta)\;.$$
  Then $M^1\sqrt{-g}=-\cos(\theta)$, and clearly
       $$\frac{\partial}{\partial x^1} (M^1\sqrt{-g})= X_1 (M^1\sqrt{-g})=0\;.$$
       Therefore $M^r$ is divergence free (recall that the term divergence is used only due to the similarity in expression).
      Since
       $$N^1\sqrt{-g}=\frac{Q \cos(\theta)}{a \sin^2(\theta)}\;,$$
       we get that
       $$\sqrt{-g}\nabla_r N^r= X_1 (N^1\sqrt{-g})=$$
       $$\frac{\cos(\theta)}{a \sin^2(\theta)}\big[(r^2+a^2) \;Q_{,t}-\Delta \;Q_{,r}\big]=0$$
       as required. The above equality follows from eq.(\ref{Qprop}).
       Since $\{M^r\}$ and $\{N^r\}$ are divergence free, we have met all the requirements of the theorem. Therefore, eq.(\ref{finalformF}) and (\ref{nulladapforms}) imply that
       $$F_{Null}= u(\theta,x^4)\; d\theta \wedge\left[dt+\frac{\rho^2}{\Delta} dr-a\sin^2(\theta) d\varphi\right]$$
        satisfies the force-free equations for any smooth, but otherwise arbitrary $u(\theta,x^4)$. Here we have rewritten  as a function of $\theta$ and $x^4$. This is an acceptable trade considering the transformation given by eq. (\ref{nulladapforms}). We have also made the substitution 
       $$\frac{u}{a \tan(\theta) \sin^2(\theta)} \rightarrow u\;.$$
       Since we have not bothered to write out the explicit coordinate transformation for the adapted chart, the above solution can seem abstract and not very recognizable. To alleviate this confusion, note that $u(\theta,x^4)$ is not to be a function of $x^1$ and $x^2$. 
       In particular, this means that $X_1 (u) =0=X_2 (u)$.
     When written in the infalling Kerr-Schild coordinates 
        $$F_{Null}= -u(\theta, x^4)\; d\theta \wedge \left[d\bar t -a\sin^2(\theta) d\bar \varphi\right]$$
       $$ = - u(\theta, x^4)\; d\theta \wedge n^\flat\;,$$
where $n^\flat$ is the $1$-form defined by
  $$(n^\flat)_\mu = g_{\mu\nu} n^\nu\;,$$
  and $u(\theta, x^4) = u(\bar t, \bar r, \theta, \bar \varphi)$ is subject to the constraints
  $$X_1 (u) = -\partial_{\bar r} \;u =0\;,$$
  and
  $$X_2 (u) = a \sin^2 \theta \;u_{,\bar t} + u_{,\bar \varphi} =0\;.$$
  In fact, the solution $F_{Null}$ as presented above is the $\bar t$ and $\bar \varphi$ dependent generalization of the original derivation in \cite{MD07}.
  This generalization was first noted in \cite{BGJ13}. In our development here it is clear why we must necessarily have this generalization. In\cite{GT14}, the authors constructed a further generalization. It should also be clear why we do not get that generalization here: particular class of solutions depend entirely on the chosen null foliation!

\subsection{$F^2 \neq 0$}
In this subsection we will require that $M\neq 0$ in $U_p$.
\begin{widetext}
\begin{theorem}
Let ${\cal F}$ be a $2$-dimensional foliation of a spacetime with metric $g$.  Let $\big(U_p, \phi_p= (x^1, \dots, x^4)\big)$ be an adapted, starlike chart centered about any point $p$. Let ${\cal F}_p$ be the maximal submanifold containing $p$ in ${\cal F}$, and let $S$ be the connected slice in ${\cal F}_p \cap U_p$ containing $p$. For $M^r$ and $N^r$ as defined above, suppose
\begin{itemize}
\item Non-null condition: $M \neq 0 \;{\rm in}\;U_p\;,$
\item Smoothness condition:
$$\left(\frac{\nabla_r M^r}{M}\right)_{,3}+\left(\frac{\nabla_r N^r}{M}\right)_{,4}=0\;,$$
and
\item Gauge conditions: $$\left(\frac{\nabla_r M^r}{M}\right)_{,a}=0=\left(\frac{\nabla_r N^r}{M}\right)_{,a}=0 \;{\rm for} \;a=1,2\;.$$
\end{itemize}
Then, there exists a  unique (up to an integration constant), smooth, force-free electromagnetic field $F$ on $U_p$ such that $\ker F$ are given by the integral submanifolds of the foliation if and only if smoothness and gauge conditions of the theorem are met.
\label{mainexistuni}
\end{theorem}
\begin{proof}As mentioned previously,  the adapted chart is such that eq.(\ref{adapchart}) is satisfied. As we have shown above, in $U_p$, the force-free electromagnetic field $F$ can be written as a simple $2$-form given by eq. (\ref{finalformF}).
Since $u$ is only a function of $x^3$ and $x^4$, clearly the gauge condition must be satisfied, and 
prescribing $F|_S$ is simply a matter of picking
$u|_S = u_0$
for some constant $u_0$. By definition, the smoothness condition on $u$ is unavoidable for the class of solutions we are looking for. Therefore, the ``only if'' portion of the theorem is proved.

\vskip0.2in
To prove the existence and uniqueness of a solution, on $U_p$, define a 1-form $\omega$ by
$$\omega = \frac{\nabla_r N^r}{M} \;dx^3 - \frac{\nabla_r M^r}{M}\; dx^4\;.$$
The smoothness condition implies that $d\omega =0$. The Poincare' lemma then implies that on $U_p$, there exists a function $\tilde u$ such that
$$\omega = d \tilde u\;.$$
The Poincare' lemma, in fact, gives a formula for the construction of the potential function as well, and is given by
$$\tilde u = \int_0 ^1 \left[\frac{\nabla_r N^r}{M} (t x^3, t x^4)\;x^3 -\frac{\nabla_r M^r}{M} (t x^3, t x^4)\;x^4\right] dt\;.$$
In the above integral, the integrands are evaluated along $(t x^3, t x^4)$ for $t\in [0,1]$. Since $U_p$ is starlike the integrands are well defined.
Then
$$\partial_3 \tilde u = \int_0 ^1 \left(\frac{\nabla_r N^r}{M}\right)_{,3} (t x^3, t x^4)\;x^3\; t\; dt +\int_0 ^1 \frac{\nabla_r N^r}{M} (t x^3, t x^4) \;dt - \int_0 ^1\left(\frac{\nabla_r M^r}{M} (t x^3, t x^4)\right)_{,3} \;x^4 t\;dt$$

$$= \int_0 ^1 \left(\frac{\nabla_r N^r}{M}\right)_{,3} (t x^3, t x^4)\;x^3\; t\; dt+\int_0 ^1 \frac{\nabla_r N^r}{M} (t x^3, t x^4) \;dt + \int_0 ^1\left(\frac{\nabla_r N^r}{M} (t x^3, t x^4)\right)_{,4} \;x^4 t\;dt\;.$$
The last term in the right hand side above was modified using the smoothness condition. Therefore
$$\partial_3 \tilde u= \int_0 ^1 \frac{d}{dt}\left(t\frac{\nabla_r N^r}{M} \right) \;dt =\frac{\nabla_r N^r}{M}\;.$$
After a similar calculation of $ \partial_4 \tilde u$, we see that $d (\ln u) = d\tilde u$. This can be solved to give
\begin{equation}
u=u_0 \; \exp\tilde u\;,   
\label{unonnull}
\end{equation}
which satisfies all the requirements of the theorem. Now suppose $u_1$ and $u_2$ are two solutions to force free equations on $U_p$ that agree on $S$, then $d(u_1 - u_2) = du_1 - du_2 =0$. I.e., $u_1 = u_2 + c$, where $c$ is a constant which must vanish since the two solutions agree on $S$.
$\blacksquare$
\end{proof}
\end{widetext}
Notice what the theorem enables us to do: In the non-null case, albeit locally, existence of force-free solutions is directly dependent on the existence of foliations where the components of the metric tensor satisfies a prescribed set of properties. The search for force-free solutions can now be a topic of study for geometers as well.

In a recent paper (\cite{Menon15}), we found an exact solution to the force-free magnetosphere of a Kerr black hole. The solution was not globally well behaved. Nonetheless, it will be instructive to reconstruct the solution using a coordinate system that is adapted to the foliation generated by the kernel of $F$.

In the present formalism, since the foliations take the primary role, let us begin by describing the Kernel of $F$. The solution is magnetically dominated, and this implies that $M \neq 0$, and the above formalism applies. We will denote this solution as $F_{Mag}$. Define vector fields in the Boyer-Lindquist coordinates of the Kerr geometry by 
$$X_1 = (r^2 +a^2) \;\partial_t + a\; \partial_\varphi\;,$$
and
$$X_2 = a \sin^2(\theta) \;\partial_t + L \sin(\theta) \;\partial_\theta + \;\partial_\varphi\;,$$
where $L=L(r)$ is of the form
$$L = f/\sqrt{C^2 - f^2}\;,$$
where $f$ is an arbitrary function of $r$, and $C$ is an integration constant.
Then (\cite{Menon15})
$$\ker F_{Mag} = {\rm span}\{X_1, X_2\}\;.$$
Since our goal is it illustrate the method by which we construct a foliation adapted coordinate system, for computational ease, we set
$L$ (and hence $f$) as a constant.  This is tantamount to the vacuum case since in the $F_{Mag}$ solution the current vector $J$ is given by (as shown in (\cite{Menon15}))
$$J_{Mag}= \frac{f_{,r}}{a \rho^2 \sin^2(\theta)}  \;X_2\;.$$
When $L,f$ is a constant, define vector fields
$$X_3 = P(t,r)\; \partial_t + (r^2+a^2)\; \partial_r$$
where $$P(t,r)=\frac{2 r a}{L} \cos(\theta)+2rt\;,$$
and for consistency and efficiency of notation, set
$$X_4 = \partial_\varphi\;.$$
It is easy to verify that
$[X_i, X_j]=0$ for $i,j=1,\dots,4\;.$
I.e., just as in the previous example, there exists an adapted coordinate system $(x^1, \dots, x^4)$ such that
\begin{equation}
   \frac{\partial}{\partial x^i}=X_i\;.
   \label{cordvect}
\end{equation}
Here the bases $1$-forms are given by
\begin{widetext}
$$\left(
    \begin{array}{c}
      dx^1 \\
      dx^2 \\
      dx^3\\
      dx^4
    \end{array}
  \right)=\frac{1}{r^2+a^2}\left(
           \begin{array}{cccc}
             1 &\frac{-P}{r^2+a^2}&\frac{-a\sin(\theta)}{L}&0 \\
             0&0&\frac{r^2+a^2}{L\sin(\theta)}&0 \\
             0&1&0&0\\
             -a&\frac{aP}{r^2+a^2}&\frac{-\rho^2}{L\sin(\theta)}&r^2+a^2
           \end{array}
         \right) \left(\begin{array}{c}
      dt \\
      dr \\
      d\theta\\
      d\varphi
    \end{array} \right)\;.$$
    \end{widetext}
In this chart,
$$\sqrt{-g} = \rho^2 \sin^2(\theta) (r^2+a^2)^2 L\;.$$ We already know that $F_{Mag}$ must take the form
$$F_{Mag} = u(x^3, x^4)\; dx^3 \wedge dx^4=$$
\begin{equation}
    u \left[\frac{dr}{r^2+a^2}\wedge\left(\frac{-a}{r^2+a^2}dt-\frac{\rho^2}{(r^2+a^2)L\sin(\theta)}d\theta +d\varphi\right)\right]\;.
    \label{FMagform}
\end{equation}
To determine the governing equations for $u$ we need to compute quantities $\{M^r\}$, and $\{N^r\}$ as defined in the previous section.
To impose eq.(\ref{FF v3_a}), we begin by calculating the components of $\{M^r\}$:
$$M^1=-\frac{(L^2+1)}{L^2} \frac{a\Delta}{\rho^2(r^2+a^2)^4}\;,$$
$$M^2=\frac{1}{L^2\rho^2\sin^2(\theta)} \frac{\Delta}{(r^2+a^2)^3}\;,$$
and
$$M = M^4 = -\frac{(L^2+1)}{L^2}\frac{\Delta}{\sin^2(\theta)(r^2+a^2)^4}\;.$$
Then from eq.(\ref{cordvect}) we get that
\begin{widetext}
$$\nabla_r M^r=\frac{1}{\sqrt{-g}}\left[X_1 (M^1\sqrt{-g})+X_2 (M^2\sqrt{-g}) + X_4 (M^4\sqrt{-g})\right]=0\;.$$
\end{widetext}

\noindent From eq.(\ref{FF v3_a}) and definition of $X_4$ we must then have that $u_{,\varphi}=0$. 
To impose eq.(\ref{FF v3_b}), we perform a similar calculation using the components of $\{N^r\}$. Here,
$$N^1=-\frac{(L^2+1)}{L^2} \frac{P\Delta}{\rho^2(r^2+a^2)^4 \sin^2(\theta)}\;,$$
$$N^2=\frac {a P \Delta}{\rho^2(r^2+a^2)^4L^2\rho^2\sin^2(\theta)}\;,$$
and of course $N^3=-M$. Just as above, a careful calculation of $\nabla_r N^r=$ yields that
$$\frac{1}{M} \nabla_r N^r = -\frac{2}{\Delta}\big[(r^2+a^2)(r-M)-2r\Delta\big]\;.$$
Since $u_{,\varphi}=0$, and
$$\frac{\partial u}{\partial x^1} =0= X_1 (u)$$
we have that $u_{,t}=0$. Therefore,
$$\frac{\partial u}{\partial x^3} =X_3 (u) = \;(r^2+a^2) u_{,r}$$
Eq.(\ref{FF v3_b}) then requires that
$$\frac{r^2+a^2}{u}\; \partial_r u = -\frac{2}{\Delta}\big[(r^2+a^2)(r-M)-2r\Delta\big]\;.$$
This is easily integrated to give that
$$u = u_0 \frac{(r^2+a^2)^2}{\Delta}\;,$$
where $u_0$ is an integration constant. Inserting the above expression into eq.(\ref{FMagform}), and as shown in \cite{Menon15}, we get that
$$F_{Mag}= \frac{u_0}{\Delta} \; dr\wedge\left[a dt +\frac{\rho^2}{ L\sin(\theta)}d\theta -(r^2+a^2) d\varphi\right]\;,$$
is a vacuum solution in Kerr geometry. Incidentally, since $F_{Mag}$ is a vacuum solution, its Hodge-Star dual, denoted by $\tilde F_{Mag}$ is also a vacuum solution, and is given by
$$\tilde F_{Mag}= \frac{u_0}{\sin(\theta)} \; d\theta \wedge \left[a \sin^2(\theta) d\varphi- dt\right] +\frac{dt}{L} \wedge d\varphi\;.$$

\subsection{Foliation by Commuting Killing Vector Fields}
In the previous theorem, one type of stumbling block arises when $\det g$ and the components $\{M^r\}$ and $\{N^r\}$ are functions of $x^1$ and $x^2$. This is because the only variable $u$ that we have is not dependent on the first two coordinates of the adapted chart. But in the event we are spared of coordinates $x^1$ and $x^2$, in the following theorem we will show that the smoothness condition is automatically satisfied. Moreover, in this case, we get an explicit expression of the electromagnetic field.
\begin{theorem} Suppose, in the adapted coordinate chart

\begin{itemize}
\item $M \neq 0\;,$
\item $\;\; \partial_a \det g=0 =\partial_a M \; {\rm for}\; a=1,2\;,$
\item $\;\;\partial_1 \;M^1 + \partial_2 \;M^2=0\;,$ and
\item $\;\;\partial_1 \;N^1 + \partial_2 \;N^2=0\;.$
\end{itemize}
Then there exists a  smooth, force-free electromagnetic field $F$ on $U_p$ such that $\ker F$ are given by the integral submanifolds of the foliation. As before, the solution is unique whenever $F|_S$ is prescribed and is given by
\begin{equation}
F = \frac{q}{M \sqrt{-g}} \;dx^3 \wedge dx^4\;,
\label{explicitsol}
\end{equation}
for some constant $q$.
\label{explicitsoltheo}
\end{theorem}
\begin{proof}
The conditions of the theorem imply that
$$\nabla_r M^r= \partial_4 M + M \partial_4 \sqrt{-g}$$
and
$$\nabla_r N^r= -\partial_3 M - M \partial_3 \sqrt{-g}\;.$$
Since we have already required the non-null condition, it is now a trivial matter to check the gauge and smoothness conditions. Finally, note that
$$d(\ln u) = -d(\ln(M\sqrt{-g}))\;.$$
This is easily integrated to give the expression stated in the theorem.
$\blacksquare$
\end{proof}
This theorem clarifies an important point. Note that  in general, one cannot expect a non-full force-free field to smoothly become null. In the limit that $M \rightarrow 0$, $F$ becomes undefined. So, in reality, the force-free condition must break down. If there are indeed  cases where a smooth limit occurs, other restricting conditions have to be met to allow this smooth transition. This is why our general formalism treats the null case separately. As is clear from eq.(\ref{explicitsol}), we are able to formulate a singularity theorem of FFE in the following way.
\begin{corollary}
When the requirements of theorem \ref{explicitsoltheo} holds, in the limiting case, the solution to FFE becomes singular as the electromagnetic field becomes null.
\label{singularity}
\end{corollary}

The conditions for the theorem above are not easily met in general, and there is no direct way to recognize them from the start save in the case when spacetime admits 2 commuting Killing vector fields. But, when this happens, there is a easy expression for at least 1 force-free solution when $M \neq 0$. The following corollary is an immediate consequence of theorem \ref{explicitsoltheo} above.
\begin{corollary}
Let $({\cal M}, g)$ admit two commuting Killing vector fields $X_1$ and $X_2$. Then in the adapted coordinate system, where $\partial_a = X_a$ for $a = 1,2$, if $M\neq 0$, there exists a non-null force-free electromagnetic solution given by eq.(\ref{explicitsol}).
\end{corollary}

\subsection{Vacuum Solutions}

In the adapted coordinate system, we already have that $(J^\sharp)^a =0$ for $a=3,4$. In the event
$$g^{\mu3}g^{\nu4} =0$$
whenever $\mu, \nu \neq 3,4$, in eq.(\ref{jexpress}), $\alpha, \beta = 1, 2$. In this case, we have that
$J=0$, and our formalism reduces to the case of vacuum solutions. A simple example in Kerr spacetime will illustrate the, albeit limited, power of eq.(\ref{explicitsol}).

\vskip0.2in 
Here $\partial_t$ and $\partial_\varphi$ are Killing vector fields, and so we set $\{x^1 = t, x^2=\varphi, x^3 = r, x^4=\theta\}$ for our adapted coordinate system. I.e., we are simply using the Boyer-Lindquist coordinate system.
In this case
$$M = -\frac{\Delta}{\rho^4}\;,$$
and so from eq.(\ref{explicitsol})
$$F_{KV}= q\frac{ \rho^2}{\Delta \sin(\theta)}\; dr \wedge d\theta$$
is easily verified to be a vacuum solution. Here, $q$ is the integration constant, and $KV$ stands for Kerr vacuum. The above expression for $F_{KV}$ is exactly the same in appearance in the horizon penetrating Kerr-Schild coordinate system as well.
As expected, $F_{KV}$ is undefined at the event horizon given by $\Delta = 0$ even though the metric is not singular in the Kerr-Schild coordinate system. In our case the solution is necessarily singular due to our singularity theorem (corollary \ref{singularity}). The metric, when restricted to the kernel of $F_{KV}$, becomes degenerate and hence the solution  approaches the null limit. The null Killing vector
$$(r_+ ^2 + a^2)\; \partial_t + a\; \partial_\varphi$$ has a vanishing inner product with every tangent vector of the kernel. 

Notice that $F_{KV}$ is one of the terms in $F_{Mag}$. But, when treated as separate solutions, their individual kernels are entirely different distributions.
As mentioned before, since $F_{KV}$ is a vacuum solution, its Hodge-Star dual given by
$$\tilde F_{KV}= q\;dt \wedge d\varphi$$
is also a vacuum solution in Kerr geometry.
\section{An example in FRW Cosmoloy}
To further illustrate the computational merits of the formalism described above, we include a new solution to the Maxwell field in a Friedman-Robertson-Walker spacetime (as far as the author is aware). For concreteness, we will pick the choice where the sectional curvature $k$ is set to $1$.
In the hyper-spherical coordinate system the metric, in this case, takes the form
$$g=-dt^2+ a^2(t)\; dr^2 +a^2(t) \sin^2 r\; d\Omega^2\;. $$
Here $ d\Omega^2$ is the metric of a unit $2$-sphere. We do not place any restrictions on the matter content of this universe and consequently $a(t)$ is not fixed. To fix the foliation, let us pick an adapted frame given by
$$X_1=\partial_\theta, \;X_2 = \partial_\varphi, \; X_3 = \partial_t, \;{\rm and}\;\;X_4 = \partial_r\;.$$
We shall denote the resulting solution as $F_{FRW}$.
By this choice, we are setting
$$\ker F_{FRW} = {\rm span}\{\partial_\theta, \partial_\varphi\}\;.$$
Clearly, $\ker F_{FRW}$ is not degenerate. Further since it is spacelike, we will see that $F_{FRW}$ is electrically dominated. Also, since the metric is diagonal in the adapted coordinate, as previously mentioned, if a solution exists, it is guaranteed to be a vacuum solution. Following eq.(\ref{finalformF}), we shall tentatively write 
$$F_{FRW}= u(t,r) \;dt \wedge dr\;.$$
A direct and trivial calculation shows that the only non-trivial components of $M^r$ and $N^r$
are given by
$$M^4=\frac{1}{a^2}=-N^3\;.$$
Force-free equations (\ref{FF v3_a}) and (\ref{FF v3_b}) reduce to the tractable form given by
$$\partial_r \ln|u| = -\partial_r \ln \sin^2 r$$
and
$$\partial_t \ln|u| = -\partial_t \ln |a|\;.$$
The above two equations are easily integrated to give
$$F_{FRW}= \frac{u_0}{a(t) \sin^2 r} \;dt \wedge dr\;.$$
It is easy to check that this solution satisfies Maxwell's vacuum equations. Additionally, the dual magnetically dominated solution in this case is given by
$$\tilde F_{FRW}= *F_{FRW}=- u_0 \;a(t) \sin^2 \theta\;d\theta \wedge d\varphi\;.$$
In both the cases, $u_0$ is the usual integration constant.
\section{Discussion And Conclusion}
We have made great progress in understanding the connection between FFE and foliation. The initial data surfaces take the form of 2-dimensional submanifolds whose tangent space agrees with the kernel of the degenerate electromagnetic field. One is not free to pick the submanifolds. Indeed, the entire difficulty in the theory of FFE is in finding the appropriate foliation of spacetime with suitable submanifolds. Once we find the foliation, a solution is guaranteed. This work suggests that the initial value problem of FFE may not be a meaningful concept. Instead, we must focus our attention on FFE admissible spacetimes and its foliations.

When $F^2 \neq 0$, the solution is unique modulo an integration constant (which is the only choice in local initial data). In the null case, we get a class of solutions depending on 2 different parameters (coordinates $x^3$ and $x^4$) of the theory.

Further, we are able to explain why the general theory of FFE without any further restrictions separates into the null case and the non-null case. This is because smooth transitions are not generally allowed. In a certain class of solutions, one that is generated by a pair of commuting Killing vector fields, solutions necessarily become undefined as one approaches the null limit.

It is important to mention that our formalism includes the Blandford-Znajek mechanism. In addition to the worked out examples, any new solution to the Blandford-Znajek mechanism must satisfy eqs. (\ref{FF v3_a}) and (\ref{FF v3_b}). This is not surprising since the equations governing $u$ are derived from the fully covariant force-free equations of electrodynamics in curved spacetime. Unlike the Blanford-Znajek mechanics however, this paper does not place any further restrictions like stationarity or axis-symmetry. Nor do we fix the background metric. Also, since the exteriors of black holes are covered by a single coordinate chart, new numerical recipes can be constructed to search for foliations with the required properties. Once the foliations are found, one can retroactively integrate to find closed-form solutions given by eq.(\ref{finalformF}) in the null case, or eq. (\ref{unonnull}) otherwise.

Clearly, this work expects a follow up in several directions. The connection between FFE and foliations have been tackled by breaking covariance. I.e., we have heavily relied on the adapted coordinate system. The recasting of the theory in completely geometric terms should now be tractable.

Further, our formalism only develops a local theory. The maximal submanifolds of the foliation could extend much further. So, it will be helpful to see how the local solutions can be smoothly stitched together. It maybe that further topological restrictions are necessary to allow such an extension.  

Finally, the issue of a smooth transitions between solutions describing null and non-null fields have to be studied in greater detail. It is not clear whether the singularity theorem we have proved can be extended to a larger class of solutions.
\section{Appendix}
\subsection{Exterior Calculus}
\label{ExtCalc}
In this section, we simply list the relevant formulae of exterior calculus we have used through out the paper. Let a differential form $\omega$ be given by
$$\omega=\omega_{i_1\;\dots\;i_k} \;dx^{i_1} \otimes \;\dots\; \otimes dx^{i_k}\;,$$
where the component functions $\omega_{i_1\;\dots\;i_k}$ are completely antisymmetric, then
$$\omega=\frac{1}{k!}\;\omega_{i_1\;\dots\;i_k} \;dx^{i_1} \wedge \;\dots\; \wedge dx^{i_k}\;.$$
The {\bf exterior derivative} of $\omega$ is defined by the expression
$$d\omega= \frac{1}{k!}\;\omega_{i_1\;\dots\;i_k,r}\;\;dx^r \wedge dx^{i_1} \wedge \;\dots\; \wedge dx^{i_k}\;.$$
In an $n$ dimensional manifold, the Hodge-Star operator $*$ takes a $k$ form to a $n-k$ form. It is defined by the formula
$$*\omega= \frac{1}{k! (n-k)!}\;(*\omega)_{i_1\;\dots\;i_{n-k}} \;dx^{i_1} \wedge \;\dots\; \wedge dx^{i_{n-k}}\;.$$
where
$$(*\omega)_{i_1\;\dots\;i_{n-k}}=\epsilon_{j_1 \dots j_k {i_1} \dots i_{n-k}}\omega^{j_1\;\dots\;j_k}\;.$$
The components of $\omega$ are raised as usual by the inverse metric tensor. I.e.,
$$\omega^{j_1\;\dots\;j_k}=g^{j_1 i_1} \dots g^{j_k i_k}\; \omega_{i_1\;\dots\;i_k}\;.$$
\vskip0.2in\noindent
The Poincare' lemma tells us that all closed forms are locally exact:
Let $U$ be an open, starlike set about a point $p$ in the manifold. Let $\omega$ be a $k$ form on $U$ such that $d\omega =0$. Then there is a $k-1$ form $\alpha$ on $U$ such that
$$\omega = d \alpha\;.$$
\vskip0.2in\noindent
Contractions of the Levi-Civita tensor can be taken using the formula

$\epsilon^{a_1 a_2 \dots a_j a_{j+1} \dots a_n} \;\epsilon_{a_1 a_2 \dots a_j b_{j+1} \dots b_n}=
$
\begin{equation}
    (-1)^s \;(n-j)! \; j!\;\delta^{[a_{j+1}} _{b_{j+1}}\;\dots \delta^{a_n]} _{b_n}\;.
    \label{levicontract}
\end{equation}
Here $s$ is the index of the metric. In the case of general relativity, for our choice of signature, $s=1$. The square brackets indicate that the indices in between have to be summed in an anti-symmetric fashion. 
\vskip0.2in\noindent
Cartan's magic formula for differential forms is given by
$${\cal L}_v F = d \;i_v F + i_v \;d F\;.$$
Here ${\cal L}_v F$ is the Lie derivative of $F$ with respect to $v$.

\subsection{Useful Results From Differential Geometry}
\label{DGAppen}
Details of proofs of all the results in this section can be found in \cite{JLee13}.

Let ${\cal M}$ be a smooth manifold of dimension $m$. Let $d < m$. ${\cal D}$ is a $d$-dimensional distribution of ${\cal M}$ if for every $ p \in {\cal M}$, ${\cal D} (p)$ is a $d$-dimensional subspace of $T_p({\cal M})$. ${\cal D}$ is a smooth distribution provided ${\cal D}$ is locally spanned by smooth vector fields. ${\cal D}$ is involutive (or completely integrable) if for any smooth vector fields $X$ and $Y$ that lie in  ${\cal D}$,  $[X,Y]$ lies in ${\cal D}$. A submanifold ${\cal N}$ of ${\cal M}$ is an integral manifold of a distribution ${\cal D}$ if the tangent space of every $p\in {\cal N}$ is given by ${\cal D}(p)$.
If ${\cal D}$ is a smooth distribution in ${\cal M}$ such that there is an integral manifold passing through every point of  ${\cal M}$, then  ${\cal D}$ is clearly involutive. 
\vskip0.2in\noindent
The following powerful theorem by {\bf Frobenius} proves the converse of the previous statement:
Let ${\cal D}$ be a $d$-dimensional  involutive distribution in ${\cal M}$. Let $p \in {\cal M}$. There there exists an embedded integral manifold of ${\cal D}$ through $p$. Additionally, there exists a cubic coordinate system $(U, \phi)$ centered at $p$ and
$$\phi = (x^1, \dots, x^m)$$ such that
$x^i =$ constant for $i = d+1, \dots, m$ are integral manifolds of ${\cal D}$.

\vskip0.2in\noindent
This final result will help us identify useful, and in our case adapted, coordinate functions. 
In an $m$-dimensional manifold, let $X_1, \dots X_k$, $k\leq m$ be  smooth, point wise linearly independent, commuting  fields on an open set  about some $p \in {\cal M}$. Then there exists a coordinate a chart $\{x^i\}_{i=1} ^m$ about $p$ such that for $i=1, \dots,k$, locally
$$X_i= \frac{\partial}{\partial x^i}\;.$$

Finally, let ${\cal F}$ be a collection of submanifolds $\{{\cal F}\}_\alpha$ of fixed dimension. Let  ${\cal F}_\alpha \cap {\cal F}_\beta = \emptyset$ for $\alpha \neq \beta$ such that $\cup_\alpha {\cal F}_\alpha$ is the entire manifold ${\cal M}$, then we say that ${\cal F}$ is a foliation of ${\cal M}$, and ${\cal F}_\alpha$ are the leaves of the foliation.

\subsection{Kerr Geometry in Boyer-Lindquist/Outgoing Kerr-Schild coordinates}
From the analysis in the main body of the paper it is clear that we need the contravariant form of the Kerr metric. In the Boyer-Lindquist coordinate system, $(t, r, \theta,\varphi)$, this takes the following form:

\begin{equation}
g^{ \mu  \nu} = \left[\begin{array}{cccc}
-\frac{\Sigma^2}{\rho^2 \Delta}&0& 0& -\frac{az}{\Delta}\\
0 &\frac{\Delta}{\rho^2}&0&0\\
0 & 0& \frac{1}{\rho^2} & 0\\
-\frac{az}{\Delta}&  0&0& \frac{\Delta-a^2 \sin^2(\theta)}{\rho^2\Delta\sin^2(\theta)}\\
\end{array}\right] \;.
\label{invgkerr}
\end{equation}

Here,
$$\rho^2 = r^2 + a^2
\cos^2\theta\;,\;\;\;\Delta = r^2 -2 M r + a^2\;,$$
$$
\Sigma^2 = (r^2 + a^2)^2 -\Delta \; a^2 \sin^2\theta\;,
$$
and
$$\sqrt{-g}=\rho^2 \sin\theta\;.$$
Here $r=r_+ = M +\sqrt{M^2-a^2}$ locates the outer event horizon.

The infalling Kerr-Schild coordinates are $(\bar t, r, \theta,\bar \varphi)$. They are related to the Boyer-Lindquist coordinates by the following relations:
$$d\bar t = dt+\frac{r^2+a^2}{\Delta}dr\;, \;\;\; {\rm and} \;\;\;\; d\bar \varphi = d\varphi+\frac{a}{\Delta}dr.$$
In the Kerr-Schild outgoing coordinates, the
metric components in the basis $\{\bar t,  r, \theta,\bar
\varphi\}$ become
\begin{equation}
 \bar g_{ \mu  \nu} = \left[\begin{array}{cccc}
z-1& 1& 0& -za\sin^2\theta\\
1& 0& 0& -a\sin^2\theta\\
0 & 0& \rho^2 & 0\\
-za\sin^2\theta & -a\sin^2\theta& 0& \Sigma^2 \sin^2\theta/\rho^2\\
\end{array}\right] \;.
\label{kerrbar}
\end{equation}

\bibliography{bibliography}

\end{document}